\documentclass[11pt,a4paper]{article}

\usepackage{cite}
\usepackage{graphicx}
\graphicspath{{figures/}}
\usepackage[english]{babel}
\usepackage{latexsym}
\usepackage{url}
\usepackage{amsmath,amsthm}
\usepackage{dsfont}
\usepackage{ifthen}
\usepackage{fancybox}
\usepackage{microtype,stmaryrd}

\usepackage[charter]{mathdesign}

\usepackage[margin=2.55cm]{geometry}

\if10     
\usepackage[mathlines]{lineno}
\newcommand*\patchAmsMathEnvironmentForLineno[1]{%
  \expandafter\let\csname old#1\expandafter\endcsname\csname #1\endcsname
  \expandafter\let\csname oldend#1\expandafter\endcsname\csname end#1\endcsname
  \renewenvironment{#1}%
     {\linenomath\csname old#1\endcsname}%
     {\csname oldend#1\endcsname\endlinenomath}}%
\newcommand*\patchBothAmsMathEnvironmentsForLineno[1]{%
  \patchAmsMathEnvironmentForLineno{#1}%
  \patchAmsMathEnvironmentForLineno{#1*}}%
\AtBeginDocument{%
\patchBothAmsMathEnvironmentsForLineno{equation}%
\patchBothAmsMathEnvironmentsForLineno{align}%
\patchBothAmsMathEnvironmentsForLineno{flalign}%
\patchBothAmsMathEnvironmentsForLineno{alignat}%
\patchBothAmsMathEnvironmentsForLineno{gather}%
\patchBothAmsMathEnvironmentsForLineno{multline}%
}
\linenumbers
\fi

\usepackage{color}
\definecolor{blu3}{rgb}{.1,.0,.4}
\usepackage{hyperref}
\hypersetup{colorlinks=true, linkcolor=blu3, urlcolor=blu3, citecolor=blu3, pdfpagemode=UseNone, pdfstartview=} 

\newtheorem{theorem}{Theorem}
\newtheorem{corollary}[theorem]{Corollary}
\newtheorem{lemma}[theorem]{Lemma}

\newtheorem{graphencoding}{Graph Encoding}

\newcommand{\NN}{\ensuremath{\mathbb N}}    
\newcommand{\UU}{\ensuremath{\mathcal{U}}}  
\newcommand{\DD}{\ensuremath{\mathcal{D}}}  
\DeclareMathOperator{\polylog}{polylog}

\DeclareMathOperator{\next}{next}
\DeclareMathOperator{\prev}{prev}
\DeclareMathOperator{\DS}{DS}
\def\dart#1#2{#1\mathord\shortrightarrow#2}

\newcommand\eps{\varepsilon}

\def\DEF#1{\textbf{\emph{#1}}}

\begin{document}
\setcounter{page}{0} 

\title{Minimum Cuts in Geometric Intersection Graphs}

\author{Sergio Cabello\thanks{Faculty of Mathematics and Physics, 
		University of Ljubljana, Slovenia, and 
		Institute of Mathematics, Physics and Mechanics, Slovenia.
		Supported by the Slovenian Research Agency (P1-0297, 
		J1-9109, J1-8130, J1-8155, J1-1693, J1-2452.). 
		Email address: sergio.cabello@fmf.uni-lj.si}
\and
	Wolfgang Mulzer\thanks{Institut f\"ur Informatik,
		Freie Universit\"at Berlin, Germany.
		Supported in part by ERC StG 757609.
                Email address: mulzer@inf.fu-berlin.de}}

\maketitle

\thispagestyle{empty}
\begin{abstract}
    Let $\mathcal{D}$ be a set of $n$ disks
	in the plane. The \emph{disk graph} $G_\mathcal{D}$ 
	for $\mathcal{D}$ is the undirected graph with vertex set
	$\mathcal{D}$ in which two disks are joined by an edge
	if and only if they intersect. The \emph{directed transmission 
	graph} $G^{\rightarrow}_\mathcal{D}$ for $\mathcal{D}$ is the
	directed graph with vertex set $\mathcal{D}$ in which there
	is an edge from a disk $D_1 \in \mathcal{D}$ to a disk 
	$D_2 \in \mathcal{D}$ if and 
	only if $D_1$ contains the center of $D_2$.

	Given $\mathcal{D}$ and two non-intersecting 
	disks $s, t \in \mathcal{D}$, we 
	show that a minimum $s$-$t$ vertex 
	cut in $G_\mathcal{D}$ or in $G^{\rightarrow}_\mathcal{D}$ 
	can be found in 
	$O(n^{3/2}\polylog n)$ expected time.
	To obtain our result, we combine an algorithm for the 
	maximum flow problem in general graphs with dynamic geometric data 
	structures to manipulate the disks.

    As an application, we consider the 
	\emph{barrier resilience problem}
	in a rectangular domain. 	In this problem, we have
    a vertical strip $S$ bounded by two 
    vertical lines, $L_\ell$ and $L_r$, and a collection 
    $\mathcal{D}$ of disks. Let $a$ be a point in $S$ above 
    all disks of $\mathcal{D}$, and
    let $b$ a point in $S$ below all disks of $\mathcal{D}$. 
    The task is to find a curve from $a$ to $b$ that lies in
    $S$ and that intersects as few disks of 
    $\mathcal{D}$ as possible. 
    Using our improved algorithm for minimum cuts
	in disk graphs, we can solve the barrier resilience
	problem in 
	$O(n^{3/2}\polylog n)$ expected time.

\end{abstract}

\paragraph{Acknowledgments.}
Parts of this work were initiated at the Fifth Annual Workshop 
on Geometry and Graphs that took place March 5--10, 2017, at 
the Bellairs Research Institute.
We thank the organizers and all participants for the 
productive and positive atmosphere.

\newpage
\section{Introduction}

Let $\DD$ be a family of $n$
(closed) disks in the plane. The \DEF{disk graph} 
$G_\DD$ for $\DD$ is the undirected graph with 
vertex set $\DD$ and edge set 
\[
	E(G_\DD) ~=~ \{ D_1D_2 \mid D_1, D_2 \in \DD, \, D_1 \cap D_2 
	\neq \emptyset \}.
\]
If the disks in $\DD$ are partitioned into 
two sets $\DD_A$ and $\DD_B$, one can also define a 
\emph{bipartite} intersection graph by considering
only the edges that come from an intersection between
a disk in $\DD_A$ and a disk in $\DD_B$.
If all disks in 
$\DD$ have the same radius, we call $G_\DD$ a 
\DEF{unit-disk graph}.
A directed version of disk graphs can be 
defined as follows: for $D \in \DD$, let 
$c_D \in D$ denote the center of $D$.
The \DEF{directed transmission graph} 
$G^\rightarrow_\DD$ is the directed graph with vertex 
set $\DD$ and edge set 
\[
	E\left( {G^\rightarrow_\DD}\right) ~=~ 
	\{ D_1 \rightarrow D_2 \mid D_1,D_2 \in \DD, \, c_{D_2} \in D_1  \}.
\] 
If we ignore the direction of the edges in $G^\rightarrow_\DD$,
we obtain a subgraph of $G_\DD$.

Unit disk graphs are often used to model ad-hoc 
wireless communication networks and sensor 
networks~\cite{GG11,zg-wsn-04,HS95}. Disks of 
varying sizes become relevant 
when different sensors cover different areas.
Moreover, general disk graphs may
serve as a tool to approach other problems; 
for example, an application to the barrier
resilience problem~\cite{kumar2007barrier} 
is discussed below. Directed transmission 
graphs model ad-hoc networks where different
entities have different power ranges~\cite{PelegR2010}.

\paragraph{Minimum \texorpdfstring{$s$-$t$}{s-t} cut in disk graphs.}
Consider a graph $G = (V, E)$ with $n$ vertices and
$m$ edges, and two non-adjacent 
vertices $s, t \in V$.
A set $X \subseteq V \setminus \{ s,t \}$
of vertices is called an \DEF{$s$-$t$ (vertex) cut} 
if $G-X$ contains no path from $s$ to $t$.
Two paths from $s$ to $t$ are \DEF{(interior-)vertex-disjoint}
if their only common vertices are $s$ and $t$.
By Menger's theorem (see, for 
example,~\cite[Section 8.2]{KorteV10}), 
the minimum size of an $s$-$t$ cut equals the 
maximum number of vertex-disjoint $s$-$t$ paths,
both in directed and in undirected graphs.
Using blocking flows, Even and Tarjan, as well as 
Karzanov~\cite{EvenT75,Karzanov73}
showed that an $s$-$t$ minimum-cut can be computed
in time $O(\sqrt{n} m)$. In the worst case, if $m = \Theta(n^2)$, 
this is $O(n^{5/2})$.
This was an improvement over the previous algorithm by Dinitz~\cite{Dinic70};
see~\cite{Dinitz06} for a great historical account
of the algorithms. In particular, the use of DFS did not appear in
his original description~\cite{Dinic70}, but it was developed by 
Shimon Even and Alon Itai and included in Even's textbook~\cite{Even79}.
The more recent $O(m^{10/7})$-time algorithm of M\k{a}dry~\cite{Madry13} gives
a better running time for sparse graphs,
i.e., for $m = o(n^{7/4})$.

The size of a minimum $s$-$t$ vertex cut in a network $G$ is a 
key estimator for its vulnerability.
Since such networks
often arise from geometric settings, it is natural to consider the
case where $G$ is a disk graph.
A particularly interesting scenario of this kind
is the \DEF{barrier resilience problem}, an 
optimization problem introduced by Kumar, Lai, 
and Arora~\cite{kumar2007barrier}. In one variant of
the problem, we are given
a vertical strip $S$ bounded by two 
vertical lines, $L_\ell$ and $L_r$, and a collection 
$\DD$ of disks. Each disk $D \in \DD$ represents a region
monitored by a sensor. Let $a$ be a point in $S$ above 
all disks of $\DD$, and
let $b$ a point in $S$ below all disks of $\DD$. 
The task is to find a curve from $a$ to $b$ that lies in
the strip $S$ and that intersects as few disks of 
$\DD$ as possible (the disks do not need to lie
inside $S$). This models the 
resilience of monitoring a boundary region with 
respect to (total) failures of the sensors.
Kumar, Lai, and Arora show that the problem 
reduces to an $L_\ell$-$L_r$ minimum-cut 
problem in the intersection graph of 
$\DD \cup\{ L_\ell,L_r\}$. We mention that
for another variant of the problem, where
the endpoints $a$ and $b$ can
lie in arbitrary locations, the complexity status
is still unknown, despite many efforts by several 
researchers~\cite{AltCaGiKn17,BeregKi09,ChanKi14,EibenL20,KormanLoSiSt18,TsengKi11}.

A variant of the problem, called \DEF{minimum shrinkage},
was recently introduced by Cabello et al.~\cite{CabelloJLM20}. 
Here, the task is to shrink some of the disks, 
potentially by different amounts, such 
that there is an $a$-$b$ curve that is disjoint 
from the interiors of all disks. The objective 
is to minimize the total amount of shrinkage.
Cabello et al.~provide an FPTAS
for the version where the path is restricted to lie inside
a vertical strip and the endpoints $a$ and $b$
are above and below all the disks. This result
is achieved by reducing the problem to a barrier 
resilience instance with $O(n^2/\eps)$ disks of 
different radii.
In contrast, when the endpoints $a$ and $b$ can
lie in arbitrary locations, 
the problem is weakly NP-hard~\cite{CabelloV20}.

\paragraph{Our Results.}
We exploit the geometric structure
to provide 
a new algorithm to find the minimum $s$-$t$ cut in disk 
graphs and directed transmission graphs
in $O(n^{3/2}\polylog n)$ expected time.
For this, we adapt the approach
of Even and Tarjan~\cite{EvenT75}, extending it with
suitable geometric data structures.
Our method is similar in spirit to the algorithm
by Efrat, Itai, and Katz~\cite{EfratIK01} for 
maximum \emph{bipartite} matching
in (unit) disk graphs. However, since our graph is not
bipartite, the structure of the graph is more complex and additional care is needed.

\section{Minimum \texorpdfstring{$s$-$t$}{s-t} Cut in Disk Graphs}
\label{sec:cut}

Let $\DD$ be a set of $n$ disks in the plane,
and let $s, t \in \DD$ be two non-intersecting disks.
We show how to compute the 
maximum number of vertex disjoint paths between
$s$ and $t$ in $G_\DD$ and in $G^\rightarrow_\DD$.
This also provides a way to find a minimum $s$-$t$ (vertex) cut.
For this, we adapt the algorithm of 
Even and Tarjan~\cite{EvenT75} to 
our geometric setting. First, we suppose 
that certain geometric primitives are available
as a black box, and we analyze the running time
under this assumption. Then,
we instantiate these primitives with 
appropriate data structures to obtain the desired result.

\subsection{Generic algorithm}
\label{sec:generic}

Let $G$ be a graph with $n$ vertices and $m$ edges, 
and let $s$ and $t$ be two non-adjacent vertices of $G$.
We want to find the maximum number of paths from $s$
to $t$ in $G$ that are pairwise vertex disjoint. 
The graph $G$ is assumed to be directed.\footnote{Otherwise,
we replace each undirected edge $uv$ by 
two directed edges $\dart uv$ and $\dart vu$.
An optimal solution to the directed instance 
directly gives an optimal solution to the
undirected case.}

\begin{figure}
	\centering
	\includegraphics[]{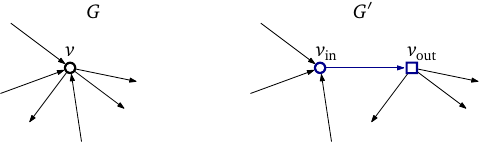}
	\caption{Transforming a vertex 
	$v\in V(G)\setminus \{s,t\}$ to get from $G$ to $G'$. }
	\label{fig:flow1}
\end{figure}

First, we transform the graph $G$ into another graph $G'$
in which every vertex other than $s$ and $t$ has 
in-degree or out-degree $1$. 
More precisely, for each vertex $v \in V(G) \setminus \{ s, t\}$, 
we perform the following operation:
we replace $v$ with two new vertices $v_{\rm in}$ and $v_{\rm out}$,
add the directed edge $\dart{v_{\rm in}}{v_{\rm out}}$, 
replace every directed edge $\dart uv$ with $\dart{u}{v_{\rm in}}$,
and replace every directed edge $\dart vw$ with $\dart{v_{\rm out}}{w}$;
see Figure~\ref{fig:flow1}.
The vertices $s$ and $t$ remain untouched.
The transformed graph $G'$ has $2n-2$ vertices and $m+n-2$ edges.
It is bipartite, as can be seen by partitioning the vertices
into  the sets
$\{ s \} \cup \{ v_{\rm out}\mid v\in V(G)\setminus\{s,t\}\}$
and $\{ t \} \cup \{ v_{\rm in}\mid v\in V(G)\setminus\{s,t\}\}$.
Vertex-disjoint $s$-$t$ paths in $G$ directly correspond to
vertex disjoint $s$-$t$ paths in $G'$. Furthermore, 
in $G'$ we have that
edge-disjoint and vertex-disjoint
$s$-$t$ paths are equivalent, because 
every vertex (other than $s$ and $t$) has in-degree or out-degree $1$.
Thus, it suffices to find the maximum number of 
edge-disjoint $s$-$t$ paths in $G'$. 

\begin{figure}
	\centering
	\includegraphics[width=\textwidth,page=2]{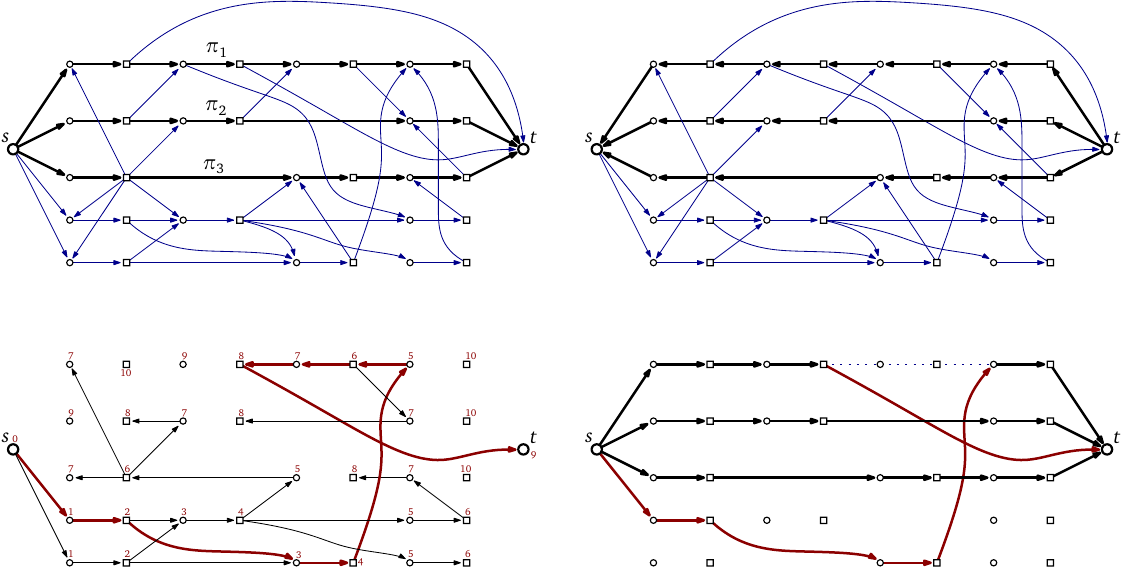}
	\caption{Top left: a graph $G'$ with $3$ vertex-disjoint $s$-$t$ 
	paths $\Pi=\{\pi_1,\pi_2,\pi_3\}$
			in bold.
			Top right: the residual graph $R(G', \Pi)$.
			Bottom left: the layered residual graph $L(G',\Pi)$.
			We keep all vertices and each vertex
			has its distance 
			from the vertex $s$ annotated. An $s$-$t$ path 
			$\gamma$ in $L$ is marked in thick red.
			Bottom right: the paths obtained 
			from $E(\Pi)\oplus E(\gamma)$.	}
	\label{fig:flow2}
\end{figure}

Assume we have a family $\Pi=\{ \pi_1,\dots,\pi_k\}$ 
of $k$ edge-disjoint $s$-$t$ paths in $G'$. 
Let $E(\Pi)=\bigcup_{\pi\in \Pi} E(\pi)$ 
denote the set of all the  
directed edges on the paths of $\Pi$.
See Figure~\ref{fig:flow2} for an illustration of the 
following concepts and discussion.
The \DEF{residual graph} $R=R(G',\Pi)$ is the directed graph
with vertex set $V(G')$ and edge set
\[
  E(R) = \{  \dart uv \mid \dart uv \in E(G')\setminus E(\Pi) \text{ or } \dart vu \in E(\Pi) \}.
\]
The residual graph $R$ is bipartite with the same bipartition as $G'$.
As in $G'$, every vertex in $V(R) \setminus \{s, t\}$ 
has in-degree or out-degree at most $1$.

For a vertex $v$ of $G'$, the \DEF{level} $\lambda(v)$ 
(with respect to $R$) of $v$ 
is the BFS-distance from $s$ to $v$ in $R$, i.e., 
the minimum number of edges on a path from $s$
to $v$ in $R$.\footnote{Recall that $R$ depends on
both $G'$ and $\Pi$.} 
If $v$ is not reachable from $s$ in $R$, 
we set $\lambda(v)=+\infty$.
For every integer $i \geq 0$,
the \DEF{layer} $L[i]$ is the set of vertices at level $i$, i.e.,
$L[i]=\{ v\in V(G')\mid \lambda(v)=i\}$.
The \DEF{layered residual graph} $L(G',\Pi)$ for $G'$ and $\Pi$
is the subgraph of the residual graph $R(G',\Pi)$
where only the directed edges from $L[i-1]$ to $L[i]$,
for $i = 1, \dots, \lambda(t) - 1$, and the 
directed edges from $L[\lambda(t)-1]$ to $t$ are kept.
More precisely, this means that
$L=L(G',\Pi)$ has vertex set $V(G')$ and directed edge set
\[
	E_t\cup \{ \dart uv\in E(R) \mid \lambda(u)+1=\lambda(v)<\lambda(t) \},
\]
where
\[
	E_t ~~=~~ \{ \dart ut\in E(R) \mid \lambda(u)+1=\lambda(t) \}.
\]

Let $\Gamma = \{\gamma_1,\dots, \gamma_\ell \}$ be
a family of edge-disjoint $s$-$t$ paths in the layered 
residual graph $L=L(G',\Pi)$. By construction,
all paths of $\Gamma$ have exactly
$\lambda(t)$ edges.
Using the $k$ paths of $\Pi$ in $G'$ and the $\ell$ paths 
of $\Gamma$ in $L$, we can obtain
$k+\ell$ edge-disjoint $s$-$t$ paths in $G'$.
For this, consider the edges
\begin{align*}
	E(\Pi)\oplus E(\Gamma) ~~&=~~
	\{ \dart uv \mid \dart uv\in E(\Pi) \text{ and } 
	\dart vu\notin E(\Gamma)\} \cup
	\{ \dart uv \mid \dart uv\in E(\Gamma) \text{ and } 
	\dart vu\notin E(\Pi)\}
\end{align*}
that are obtained from $E(\Pi)\cup E(\Gamma)$ 
by canceling out directed edges that appear in both directions.
The following observation is simple:

\begin{lemma}
\label{lem:combination}
The set $E(\Pi)\oplus E(\Gamma)$ consists of $k+\ell$ edge-disjoint 
	$s$-$t$ paths in $G'$.
	Given $\Pi$ and $\Gamma$, we can construct $E(\Pi)\oplus E(\Gamma)$ and 
	the corresponding $k+\ell$ edge-disjoint $s$-$t$ paths in $G'$
	in $O(|E(\Pi)|+|E(\Gamma)|)$ total time.
\end{lemma}

\begin{proof}
	The 
	definition of $R$ ensures that  
	the edges $E(\Pi)\oplus E(\Gamma)$ all lie in $G'$,
	since for $\dart uv \in E(\Gamma) \setminus E(\Pi)$,
	we must have $\dart vu \in E(\Pi)$. Furthermore,
	every vertex $v$ of $V(G')\setminus \{ s,t \}$ has 
	in-degree and out-degree both $0$ or both $1$ 
	in $E(\Pi)\oplus E(\Gamma)$.  This is clear
	if $v$ appears on at most one path in $\Pi \cup \Gamma$. 
	If $v$
	appears on both a path from $\Pi$ and from $\Gamma$,
	then one incoming edge and one outgoing edge of $v$
	must cancel, since $v$ has at most one incoming
	or outgoing edge in $L$ and the corresponding reverse 
	edge must have appeared on a path in $\Gamma$.
	The in-degree of $s$ is $0$ and the out-degree of $t$ is $0$.
        Moreover, the out-degree of $s$ is $k+\ell$, because 
	the outgoing edges from $s$ never cancel out.
	 This means 
	that $E(\Pi)\oplus E(\Gamma)$ defines $k+\ell$ paths from 
	$s$ to $t$. These paths can be found
	in $O(|E(\Pi)|+|E(\Gamma)|)$ time by constructing the graph 
	$(V(\Pi\cup \Gamma), E(\Pi)\oplus E(\Gamma))$ explicitly.
\end{proof}

A family $\Gamma$ of $s$-$t$ paths in the layered residual graph $L$ 
is \DEF{blocking} if $L - E(\Gamma)$ contains no $s$-$t$ path,
i.e., every $s$-$t$ path in $L$ contains at least one edge
from $E(\Gamma)$.
Even and Tarjan~\cite{EvenT75} describe the following algorithm
for finding a blocking family $\Gamma$ of $s$-$t$ paths in 
a layered residual graph $L$:
we start with $\Gamma=\emptyset$, $D_0=L$, and $j=1$. 
The algorithm proceeds in rounds. In round $j$, 
we perform a DFS traversal from $s$ in $D_{j-1}$.
When we reach $t$, the DFS stack contains a path $\gamma_j$
from $s$ to $t$ in $D_{j-1}\subseteq L$. 
We add the path $\gamma_j$ to $\Gamma$,
and we obtain $D_j$ by removing from $D_{j-1}$ 
all the vertices (other than $s$ and $t$) 
that have been explored during 
the partial DFS traversal of $D_{j-1}$.
We finish when the graph $D_{j-1}$ of the current
round $j$ does not 
contain any $s$-$t$ path.
This is detected during the DFS traversal of $D_{j-1}$.	
We refer to the paper of 
Even and Tarjan~\cite{EvenT75} 
for the running time analysis and the proof of correctness.
The following lemma summarizes the result.

\begin{lemma}[Even and Tarjan~\cite{EvenT75}]

\label{lem:oneiteration}
	Let $L$ be a layered residual graph.
	In $O(|E(L)|)$ time, we can find a blocking family 
	$\Gamma$ of $s$-$t$ paths in $L$.	
\end{lemma}

The algorithm to find the maximum number of 
edge-disjoint $s$-$t$ paths in $G'$ is the following:
we start with $\Pi_0=\emptyset$.
Then, for $j=1,\dots$, we 
construct the residual graph $R_j=R(G,\Pi_{j-1})$,
the layered residual graph $L_j=L(G,\Pi_{j-1})$, 
a blocking family $\Gamma_j$ of $s$-$t$ paths in $L_j$,
and we set $\Pi_j$ to the set of (edge-disjoint) $s$-$t$ paths
defined by $E(\Pi_{j-1})\oplus E(\Gamma_j)$.
We finish when $L_j$ contains no $s$-$t$ path.
The work performed for a single value of $j$ 
(constructing $L_j$, $R_j$, $\Gamma_j$ and $\Pi_j$),
is called a \DEF{phase}.
Let $\lambda_j(\cdot)$ denote the level of a vertex 
in the residual graph $R_j$. Even and Tarjan~\cite{EvenT75} show
that $\lambda_j(t)$ increases monotonically as a function
of $j$.  Thus, using that the paths $\Gamma_j$
are vertex-disjoint and have length $\lambda_j(t)$ (whenever $L_j$ contains
some $s$-$t$ path), one obtains the following.

\begin{theorem}[Even and Tarjan~\cite{EvenT75}]
\label{thm:EvenT75}
	The algorithm performs at most $O(\sqrt{n})$ phases. 
	When the algorithm finishes,
	$\Pi_{j-1}$ contains the maximum possible number 
	of vertex-disjoint $s$-$t$ paths in $G'$.
\end{theorem}

\subsection{Adaptation for neighbor queries}

We want to adapt the algorithm from Section~\ref{sec:generic}
to our geometric setting.
For this, we extend the approach by 
Efrat, Itai, and Katz~\cite{EfratIK01} for finding maximum matchings
in bipartite geometric intersection graphs.
The idea is to avoid the explicit construction
of the layered residual graphs $L_j = L(G, \Pi_{j-1})$, and
to use instead an implicit representation that allows
for an efficient DFS traversal of the current $L_j$.
For this, we identify which vertices belong to each layer
of the current $L_j$,
and we use dynamic nearest-neighbor data structures to find 
the directed edges between the layers.
In order to encapsulate the geometric primitives,
we assume that we have
a certain geometric data structure to access the directed edges of $G$.
Note that the assumption is on the original graph $G$, 
not in the transformed graph $G'$. Later, we will describe how
such a data structure can be derived from known results
about (semi-)dynamic nearest neighbor searching.

\begin{graphencoding}
\label{ds:encoding edges}
	Let $G$ be a directed graph with $n$ vertices.
	We assume that we have a data structure $\DS=\DS(U)$ 
	that semi-dynamically maintains
	a subset $U\subseteq V(G)$ with the following operations:
	\begin{itemize}
		\item construct the data structure $\DS(U)$ for an initial
		  subset $U \subseteq V(G)$ of vertices from $G$. The
		  construction time is denoted by $T_c(m)$, where
		  $m$ is the number of vertices in $U$, and we require
		  that $T_c(\cdot)$ satisfies 
		  $T_c(m) + T_c(m') \leq T_c(m + m')$, for all $m, m' \in \NN$;
		\item delete of a vertex $u \in U$ from $\DS(U)$.
		The deletion time is denoted by $T_d(n)$, where $n$
		refers to the number of vertices in $G$; and
		\item given a query vertex $v\in V(G)$ from $G$,
		find an outgoing edge 
		$\dart{v}{u}$ with $u \in U$, 
		or report that no such vertex exists in the current set $U$.
		The query time is denoted by $T_q(n)$,
		where $n$ refers to the number of vertices in $G$.
	\end{itemize}
\end{graphencoding}

Henceforth, we assume our $n$-vertex graph $G$ can be accessed 
as in 
Graph Encoding~\ref{ds:encoding edges}. As before,
we denote the corresponding transformed
graph by $G'$. First, we show how to find the
levels in the layered residual graph.
	
\begin{lemma}
\label{lem:levels}
	Let $\Pi$ be  a set of edge-disjoint paths 
	in the transformed graph $G'$.
	In time $O(T_c(n) + n T_q(n) + n T_d(n))$, we can
	find the level $\lambda(v)$ 
	of each vertex $v \in V(G')$ in the layered 
	residual graph $L=L(G',\Pi)$.
\end{lemma}
\begin{proof}
	Our goal is to perform a BFS in the residual graph $R = R(G', \Pi)$ 
	without explicitly constructing the edge set of $R$.
	In a preprocessing phase, for every vertex $v$ in 
	$V(G') \setminus \{s, t\}$ that appears
	in some path of $\Pi$, we mark $v$ and store the unique 
	vertices 
	$\prev(v)$ and $\next(v)$ such that	
	$\dart{\prev(v)}{v}$ and $\dart{v}{\next(v)}$ are 
	directed edges in $E(\Pi)$.
	This takes time $O(|E(\Pi)|) = O(n)$.
	
	Next, we set $L[0]=\{ s \}$, construct the data structure $\DS$ of 
	Graph Encoding~\ref{ds:encoding edges} for $V(G) \setminus \{s\}$. 
	Thus, the current vertex
	set $U$ in $\DS$ is initially
	$U = V(G)\setminus \{ s\}$.
	In our algorithm, we iteratively compute the layers 
	$L[i]$, for $i=1,2,\dots$.
	In the process, we maintain the invariant that, after computing $L[i]$,
	the structure $\DS$ contains $t$ and 
	the vertices $u$ in $V(G)$ for which we do not 
	yet know the level $\lambda(u_{\rm in})$ in
	$L(G', \Pi)$.
	
	\begin{figure}
		\centering
		\includegraphics[page=3]{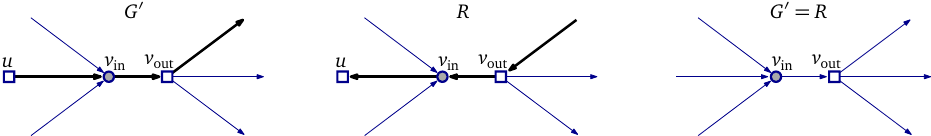}
		\caption{Case in the proof of Lemma~\ref{lem:levels}:
				$v_{\rm in}\in L[i-1]$, for $i$ even.
				The left and center figure show $G'$ and $R$ 
				when $v_{\rm in}$ belongs
				to some path of $\Pi$ (bold).
				The right figure shows $G'=R$ when 
				$v_{\rm in}$ does not belong
				to any path of $\Pi$.}
		\label{fig:flow3}
	\end{figure}
	\begin{figure}
		\centering
		\includegraphics[page=4,width=\textwidth]{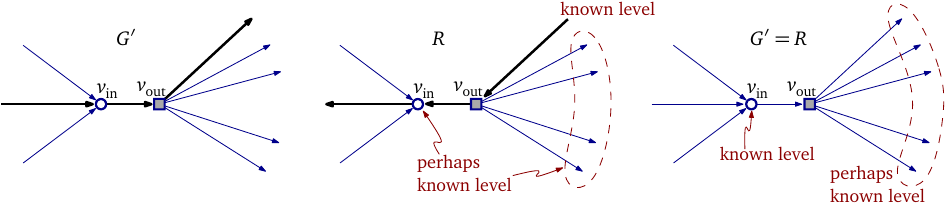}
		\caption{Case in the proof of Lemma~\ref{lem:levels}:
				$v_{\rm out}\in L[i-1]$, for $i$ odd.
				The left and center figure 
				show $G'$ and $R$ when $v_{\rm out}$ belongs
				to some path of $\Pi$ (bold).
				The right figure shows $G'=R$ when 
				$v_{\rm out}$ does not belong
				to any path of $\Pi$.}
		\label{fig:flow4}
	\end{figure}

	To find $L[1]$, we repeatedly query $\DS$ with $s$ 
	and remove from $\DS$ the reported item, until $\DS$
	contains no further out-neighbors of $s$. This gives the set
	\[
	U'=\{ u \in V(G)\setminus \{ s\}\mid \dart su\in E(G) \}
	\]
	of all out-neighbors of $s$ in $G$.
	Let $U'_{\rm in}= \{ u_{\rm in}\mid u \in U\}$ be the 
	set of corresponding out-neighbors of $s$ in $G'$.
	We filter $U'_{\rm in}$ and remove those vertices $v$ that are 
	in some path of $\Pi$ and have $\prev(v) = s$.
	This gives a set $U''_{\rm in}$ with $L[1]= U''_{\rm in}$.
	For each vertex $u_{\rm in}\in U''_{\rm in}$,
	we set $\lambda(u_{\rm in})=1$.
	For each vertex $u_{\rm in}\in U'_{\rm in}\setminus U''_{\rm in}$,
	the level of $u_{\rm in}$ in $L$ is not yet known.
	If $\DS$  supported insertions,
	we would insert	the vertices $u$ with 
	$u_{\rm in}\in U'_{\rm in}\setminus U''_{\rm in}$ back into $\DS$.
	Instead, we just construct the data structure $\DS$
	\emph{anew} for $V(G)\setminus (\{ s\}\cup \{ u\mid u_{\rm in}\in U''_{\rm in}\})$.
	
	Then, for $i = 2,\dots$, while $L[i-1]$ is not empty and $L[i-1]$ 
	does not contain $t$, 
	we compute $L[i]$.
	If $i$ is even, we iterate over the vertices $v_{\rm in}$ of $L[i-1]$;
	see Figure~\ref{fig:flow3}.
	The vertex $v_{\rm in}$ has one outgoing edge in $L$:
	if $v_{\rm in}$ does not lie on some path of $\Pi$,
	then $L$ contains only the outgoing edge 
	$\dart {v_{\rm in}} v_{\rm out}$; if $v_{\rm in}$ lies
	on some path of $\Pi$, then $L$ contains only the outgoing edge
	$\dart {v_{\rm in}} \prev(v_{\rm in})$.
	If $v_{\rm in}$ does not belong to any path of $\Pi$,
	we set $\lambda(v_{\rm out}) = i$ and add $v_{\rm out}$ to $L[i]$.
	(In this case, the only incoming edge to $v_{\rm out}$ in 
	the residual graph is from $v_{\rm in}$, 
	so we know that $\lambda(v_{\rm out})$ was not yet determined.)
	If $v_{\rm in}$ belongs to some path of $\Pi$, we 
	set $u=\prev(v_{\rm in})$ and distinguish two cases. 
	If $u=s$, we do not need to do anything because $\lambda(s)$ is already set. 
	If $u\neq s$, we set $\lambda(u)=i$ and add $u$ to $L[i]$.
	(In this case, $u = w_{\rm out}$ for some vertex $w\in V(G)\setminus \{ s,t\}$
	and $\lambda(u)$ was not yet determined because $\dart{v_{\rm in}}w_{\rm out}$ 
	is the only incoming edge to $w_{\rm out}$ in the residual graph.)
	
	If $i$ is odd, we 
	iterate over the vertices $v_{\rm out}$ of $L[i-1]$;
	see Figure~\ref{fig:flow4}.
	If the vertex $v_{\rm out}$ does not lie on some path of $\Pi$,
	the outgoing edges of $v_{\rm out}$ in $L$ correspond
	to the outgoing edges of $v_{\rm out}$ in $G'$; if 
	$v_{\rm out}$ lies on some path of $\Pi$, then the outgoing
	edge $\dart {v_{\rm out}} \next(v_{\rm out})$ in $G'$ is
	replaced with the outgoing edge $\dart {v_{\rm out}} v_{\rm in}$ in $R$.
	We proceed as follows: we query $\DS$ 
	repeatedly with $v$ and delete the reported items.
	This gives the set $U'$ of vertices $u \in V(G)$ 
	that are stored in $\DS$ and have $\dart vu \in E(G)$.
	Due to the invariant, the set $U'$ contains exactly those out-neighbors
	$u$ of $v$ in $G$ such that $\lambda(u_{\rm in})$ 
	was not known before processing $v_{\rm out}$. 
	If $v_{\rm out}$ lies on some path of $\Pi$, 
	then we already know the level of $w_{\rm in} = \next(v_{\rm out})$ 
	(it is $i-2 $) because $\dart{w_{\rm in}}{v_{\rm out}}$ is the only incoming
	edge to $v_{\rm out}$ in the residual graph, and
	therefore $w\notin U'$.
	For each $u\in U'$, we set $\lambda(u_{\rm in})=i$ and add 
	$u_{\rm in}$ to $L[i]$.	
	If $v_{\rm out}$ belongs to some path of $\Pi$,
	we check if $v_{\rm in}$ still has no level assigned,
	and if so, we set $\lambda(v_{\rm in})=i$,
	add $v_{\rm in}$ to $L[i]$, and delete $v$ from $\DS$.

	We finish when $t\in L[i]$ or when $L[i]$ is empty. 
	In the latter case, $t$ cannot be reached from $s$ in $R$, 
	and therefore $\Pi$ already contains a maximum number 
	of vertex-disjoint $s$-$t$ paths. In the former case, 
	we remove all elements from $L[\lambda(t)]$ except for $t$. 
	
	To bound the running time, we note first that
	it takes $O(T_c(n))$ time to construct the 
	data structure $\DS$, and this is done twice.
	Next, we observe that every node $u$ of $G$ is 
	deleted at most once from $\DS$.
	Additionally, each query with a vertex of $G$ in $\DS$ 
	leads either to a deletion in $\DS$
	or does not yield an out-neighbor of the vertex, but the latter happens
	at most once per vertex of $G$.
	Thus, in total we are making $O(n)$ queries and deletions in the 
	data structure $\DS$. The time bound follows.
\end{proof}

The next lemma shows how to find an actual blocking family in
$L$.

\begin{lemma}
\label{lem:oneiterationgeometric}
	Consider a set $\Pi$ of edge-disjoint paths in $G'$.
    In $O(T_c(n)+ n T_q(n) + n T_d(n))$ time, we can
	find a blocking family of $s$-$t$ paths in the 
	layered residual graph $L$.
\end{lemma}
\begin{proof}
	Using Lemma~\ref{lem:levels}, we compute the 
	level $\lambda(v)$ of each vertex $v$ of $G'$. 
	Recall the notation $\prev(v)$ and $\next(v)$
	from the proof of Lemma~\ref{lem:levels} 
	to denote the predecessor and successor
	of a vertex $v$ on a path of $\Pi$.
	We adapt the algorithm in the proof
	of Lemma~\ref{lem:oneiteration}, which is based 
	on a DFS traversal of $L$.
	
	For each odd $i$ with $1 \le i \le \lambda(t)$, we 
	build a data structure $\DS[i]$ as in Graph 
	Encoding~\ref{ds:encoding edges} for the set 
	$V[i] = \{ v\in V(G)\mid v_{\rm in}\in L[i]\}$.
	This takes $\sum_i T_c(|V[i]|) \le O(T_c(n))$ time 
	because the sets $V[i]$ are pairwise disjoint.
	During the algorithm, the data structure $\DS[i]$ will 
	contain the vertices $v \in V[i]$ 
	such that $v_{\rm in}$ has not yet been explored by the 
	DFS traversal. Thus, in contrast to the approach in 
	Lemma~\ref{lem:oneiteration},
	we delete vertices as we explore them with the DFS traversal.

	When we explore a vertex $v_{\rm in}$ (at odd level $i$), 
	there are two options; see Figure~\ref{fig:flow3}.
	If $v_{\rm in}$ lies on some path of $\Pi$, 
	we look at $u=\prev(v_{\rm in})$. If $u$ has been explored already,
	we return\footnote{This happens only if $u=s$, as in any other case 
	$u = w_{\rm out}$ for some vertex $w\in V(G)\setminus \{ s,t\}$
	and $\dart{v_{\rm in}}w_{\rm out}$ is the only incoming edge to 
	$w_{\rm out}$ in the residual graph and thus in the layered residual graph.}.
	Otherwise, we continue the DFS traversal at $u$.
	If $v_{\rm in}$ does not belong to any path of $\Pi$,
	then $v_{\rm out}$ has not been explored yet, as 
	$\dart{v_{\rm in}}{v_{\rm out}}$
	is the only incoming edge of $v_{\rm out}$,
	so we continue the DFS at $v_{\rm out}$.
	For each such vertex, we spend $O(1)$ time plus the time
	for the recursive calls, if they occur.
	
	Consider now the case that we explore a vertex $v_{\rm out}$, 
	at even level $i$; see Figure~\ref{fig:flow4}.
	If $i = \lambda(t)-1$, we check
	whether the edge $\dart vt$ belongs to $G\setminus E(\Pi)$. 
	If so, we have found an $s$-$t$
	path $\gamma$ in $L$. We add $\gamma$ to the output,
	and restart the DFS traversal from $s$.
	If not, we return from the recursive call.
	
	Consider the remaining case: we explore a vertex $v_{\rm out}$
	at even level $i$ and $i<\lambda(t)-1$.
	If $v_{\rm out}$ belongs to some path of $\Pi$,
	$v_{\rm in}$ has not been explored yet, and 
	$\lambda(v_{\rm in}) = \lambda(v_{\rm out}) + 1$,\footnote{In the
	journal version of this article, this third condition
	is missing. However, it is necessary for the algorithm to
	be correct. We thank Matej Marinko for pointing this out.}
	we recursively explore $v_{\rm in}$ and remove $v$ from $\DS[i+1]$.
	If $v_{\rm out}$ does not belong to any path of $\Pi$
	or we have returned from the exploration of $v_{\rm in}$,
	we explore the outgoing edges from $v_{\rm out}$ 
	to $L[i+1]$ by repeating the following procedure.
	We query $\DS[i+1]$ with $v$ 
	to obtain an edge $\dart vu$ of $G$ such that 
	$u_{\rm in}\in L[i+1]$, we remove $u$ from $\DS[i+1]$,
	and we continue the DFS traversal from $u_{\rm in}$.
	The recursive call is correctly made along an edge of the layered residual graph
	because it \emph{cannot} happen that $\dart{v_{\rm out}}{u_{\rm in}}$ is an edge of $\Pi$;
	indeed, if $\dart{v_{\rm out}}{u_{\rm in}}$ were an edge in $\Pi$,
	then in the residual graph the edge $\dart{u_{\rm in}}{v_{\rm out}}$
	would be the only edge incoming into $v_{\rm out}$, which would mean that
	in the DFS traversal we arrived to $v_{\rm out}$ from $u_{\rm in}$,
	and $u$ would belong to $V[i-1]$ instead of $V[i+1]$.
	When the query to $\DS[i+1]$ with $v$ returns an empty answer,
	we return from the recursive call at $v_{\rm out}$.
	
	Every vertex $u$ of $V[i]$, for $i$ odd, is returned and removed from $\DS[i]$
	at most once.
	Thus, each vertex of $V(G)$ is deleted exactly once from
	exactly one data structure $\DS[i]$.
	Furthermore, for every vertex $v$ of $V(G)$, we make 
	at most one query to the corresponding data 
	structure $\DS[\cdot]$ that returns an empty answer.
	Thus, the running time is $O(n+ T_c(n) + n T_q(n) + n T_d(n))$.
\end{proof}

The following lemma discusses how to find a minimum cut from
a maximum family of $s$-$t$ vertex disjoint paths.

\begin{lemma}
\label{lem:cut}
	Let $\Pi$ be a maximum family of $s$-$t$ vertex disjoint 
	paths (in $G$ or in $G'$). Given $\Pi$, we 
	can obtain a minimum $s$-$t$ cut in $O(T_c(n)+ n T_q(n) + n T_d(n))$ time.
\end{lemma}
\begin{proof}
	Consider the residual graph $R=R(G',\Pi)$.
	Let $A$ be the set of vertices in $V(G)$ that in the residual graph
	$R$ are reachable from $s$.
	A standard result from the theory of maximum flows
	tells that the edges from $A$ to $V(G)\setminus A$,
	denoted by $\delta_R(A)$,
	form a minimum edge $s$-$t$ cut in $G'$ and 
	there are $|\Pi|$ edges in such a cut $\delta_R(A)$.
	
	Let	$U$ be the set of vertices $u \in V(\Pi)$ such that
	$u_{\rm out} \notin A$ but $u_{\rm in}\in A$
	or such that 
	$u_{\rm in} \notin A$ but $\prev(u_{\rm in})\in A$.
	(Here, like in previous proofs, we use $\prev(u)$ to denote the vertex
	such that $\dart{\prev(u)}{u}$ belongs to $E(\Pi)$.)
	Each edge of the cut $\delta_R(A)$ contributes one
	vertex to $U$.
	Then $U$ is a minimum $s$-$t$ cut in $G$.
	
	If $t$ is not reachable from $s$ in $R$, then 
	a vertex $u$ is reachable from $s$ in the residual graph $R$ 
	if and only if $u$ is reachable from $s$ in layered residual graph $L$.
	Thus, to compute  $U$, we apply
	Lemma~\ref{lem:levels} to find
	the level $\lambda(v)$ of every vertex $v$
	in $L$. Then, the set $U$ is
	\[
		\{ u\in V(G)\mid \lambda(u_{\rm in})<+\infty, 
		\lambda(u_{\rm out})=+\infty\} \cup 
			\{ u\in V(G)\mid \lambda(\prev(u_{\rm in}))<+\infty, 
		\lambda(u_{\rm in})=+\infty\},
	\]
	as desired.
\end{proof}

Now, we put everything together.
By Theorem~\ref{thm:EvenT75}, we have $O(\sqrt{n})$ phases,
and each phase can be implemented in $O(T_c(n)+ n T_q(n) + nT_d(n))$ time 
because of Lemma~\ref{lem:combination} and 
Lemma~\ref{lem:oneiterationgeometric}.

\begin{theorem}
\label{thm:generic}
	Let $G$ be a directed graph with $n$ vertices
	and assume that a representation of its edges as given in 
	Graph Encoding~\ref{ds:encoding edges} exists.
	Then, we can find in $O(n^{1/2}(T_c(n)+ n T_q(n) + n T_d(n)))$ 
	time the maximum number of vertex-disjoint $s$-$t$ paths 
	for any given $s,t\in V(G)$.
	Similarly, we can find a minimum $s$-$t$ cut.	
\end{theorem}

\begin{proof}
	We use the algorithm described in Section~\ref{sec:generic}, 
	before Theorem~\ref{thm:EvenT75}.
	Because of Theorem~\ref{thm:EvenT75}, we have $O(\sqrt{n})$ phases.
	At phase $j$, we have a set $\Pi_{j-1}$ of vertex-disjoint paths 
	in $G'$, and we use Lemma~\ref{lem:oneiterationgeometric} 
	to find a blocking family $\Gamma_j$ of $s$-$t$ paths in 
	the layered residual graph $L_j=L(G',\Pi_{j-1})$. This takes 
	$O(T_c(n) + n T_d(n) + n T_q(n))$ time per phase.
	Because of Lemma~\ref{lem:combination}, we can then obtain 
	the new family of $s$-$t$ paths $\Pi_j$ in $O(n)$ time per phase. 
	The result for maximum number of vertex-disjoint $s$-$t$ paths follows.	
	For the minimum $s$-$t$ cut, we use Lemma~\ref{lem:cut}.
\end{proof}

\section{Geometric Applications}

Theorem~\ref{thm:generic} leads to several consequences for 
geometrically defined graphs,
as we can use geometric data structures to realize
Graph Encoding~\ref{ds:encoding edges} efficiently.
For unit disk graphs, there is the semi-dynamic data structure
of Efrat, Itai, and Katz~\cite{EfratIK01}.
The construction takes $O(n\log n)$ time,
while each deletion and neighbor query takes $O(\log n)$ amortized time.
For arbitrary disks, we can use
the structure of Kaplan et al.~\cite{KaplanMRSS17}.

\begin{corollary}
\label{cor:disks}
	Let $\UU$ be a set of $n$ unit disks in the plane
	and let $s$ and $t$ be two of the disks.
	We can find in $O(n^{3/2}\log n)$ time 
	the minimum $s$-$t$ cut	in the intersection graph $G_\UU$.
	For arbitrary disks, the running time becomes 
	$O(n^{3/2}\log^{11} n)$ in expectation.
\end{corollary}

We can easily adapt the algorithm to the case
where $s$ and $t$ are arbitrary shapes (and the other vertices
are still represented as disks), by precomputing
the disks that intersect $s$ and the disks that
intersect $t$. We get the following consequence.

\begin{corollary}
\label{cor:barrierunitdisks}
	The barrier resilience problem with $n$ unit
	disks can be solved in $O(n^{3/2}\log n)$ time. 
	For arbitrary disks, the running time becomes $O(n^{3/2}\log^{11} n)$. 
\end{corollary}

For directed transmission graphs, we can use the data structure of
Chan~\cite{Chan2019} to report a disk center contained in a query disk.
It takes $O(\log^4 n)$ amortized time per edition and query.
(See~\cite{AgarwalM95,Chan10,ChanT16,KaplanMRSS17,Liu20} for related bounds
and for an alternative presentation of Chan's data structure.)

\begin{corollary}
\label{thm:reachability}
	Let $\UU$ be a set of $n$ disks of arbitrary radii in the plane
	and let $s$ and $t$ be two of the disks.
	We can find in $O(n^{3/2}\log^4 n)$ time 
	the minimum $s$-$t$ cut
	in the directed transmission graph $G^\rightarrow_\UU$.
\end{corollary}

Similar results can be obtained for squares and rectangles using
data structures for orthogonal range searching. We next provide
concrete running times for future reference.
For intersection graphs of unit squares, 
we can again use the semi-dynamic data structure
of Efrat, Itai, and Katz, that also applies for 
the $L_1$-metric~\cite[Remark~5.5]{EfratIK01}.
As before, the construction takes $O(n\log n)$ time,
while each deletion and neighbor query takes $O(\log n)$ amortized time.

\begin{corollary}
\label{cor:squares}
	Let $\UU$ be a set of $n$ unit axis-parallel squares in the plane,
	and let $s$ and $t$ be two of the squares.
	We can find in $O(n^{3/2}\log n)$ time 
	the minimum $s$-$t$ cut	in the intersection graph $G_\UU$.
\end{corollary}

The following corollary covers also the case of squares of different sizes.

\begin{corollary}
\label{cor:rectangles}
	Let $\UU$ be a set of $n$ axis-parallel rectangles in the plane
	and let $s$ and $t$ be two of the rectangles.
	We can find in $O(n^{3/2}\log^2 n)$ time 
	the minimum $s$-$t$ cut	in the intersection graph $G_\UU$.
\end{corollary}
\begin{proof}
	To apply Theorem~\ref{thm:generic}, we need a data structure that 
	maintains a set $\UU$ of $n$ axis-parallel rectangles 
	under deletions and answers the following type of queries: 
	given an axis-parallel rectangle $R$, 
	report some rectangle $\UU$ that intersects $R$.

	For this task, Edelsbrunner~\cite{Edels80} provides a 
	data structure with 
	$T_c(n)=O(n\log^2 n)$,
	$T_q(n)=O(\log^2 n)$ amortized, and $T_d(n)=O(\log^2 n)$ amortized.	

	An alternative approach to get the desired data structure is 
	the following. Two rectangles $R$ and $R'$ intersect if and only
	if some edge of $R$ intersects some edge of $R'$, or some vertex of $R$
	is contained in $R'$, or some vertex of $R'$ is contained in $R$.
	Each one of this conditions can be checked using orthogonal 
	range searching 
	techniques, namely using range trees~\cite[Section 5.3]{BergCKO08},
	interval trees~\cite[Section 10.1]{BergCKO08} 
	and segment trees~\cite[Section 10.3, Exercise 10.8]{BergCKO08}.
	In the static setting, this readily gives a data structure
	with construction time $O(n\log^2 n)$ and worst-case query 
	time $O(\log^2 n)$.
	These data structures can handle deletions by marking certain 
	information
	as deleted and updating the pointers locally, without rebalancing.
	Together, this gives a data structure with 
	$T_c(n)=O(n\log^2 n)$, $T_q(n)=O(\log^2 n)$ in the worst case, 
	and $T_d(n)=O(\log^2 n)$ in the worst case; here $n$ denoted the 
	original number of rectangles.
	
	Using any of the two data structures and Theorem~\ref{thm:generic},
	the result follows.	
\end{proof}

The barrier problem with axis-parallel squares or rectangles can now be solved
similarly. 
For the case of unit squares, it pays off to precompute the squares intersected
by each boundary of the strip. For arbitrary squares or rectangles, we could
treat each boundary as a rectangle.

\begin{corollary}
\label{cor:barriersquares}
	The barrier resilience problem with $n$ 
	axis-parallel squares of unit side length 
	can be solved in $O(n^{3/2}\log n)$ time. 
	For arbitrary squares or rectangles, the 
	running time becomes $O(n^{3/2}\log^2 n)$. 
\end{corollary}

The reachability graph $G^\rightarrow_\UU$ can be defined also for 
sets $\UU$ of axis-parallel squares: 
there is a directed edge from square $S$ to square $S'$ if $S$ 
contains the center of $S'$.

\begin{corollary}
\label{thm:reachabilitysquare}
	Let $\UU$ be a set of $n$ axis-parallel squares in the plane
	and let $s$ and $t$ be two of the squres.
	We can find in $O(n^{3/2}\log^2 n)$ time 
	the minimum $s$-$t$ cut
	in the directed transmission graph $G^\rightarrow_\UU$.
\end{corollary}
\begin{proof}
	Use the semi-dynamic data structure to report a point contained
	in a query square that is based on 
	range trees~\cite[Section 5.3]{BergCKO08}, as sketched in the 
	proof of Corollary~\ref{cor:rectangles}.
\end{proof}

\section{Conclusion}
We have shown how to combine the classic maximum-flow
algorithm of Even and Tarjan~\cite{EvenT75} with recent geometric data 
structures in order to find a minimum $s$-$t$ cut in geometric
intersection graphs. Even though we follow along the lines of 
the classic algorithms, the details for an efficient implementation
in the geometric setting are quite subtle and show an interesting
interplay between geometric and combinatorial algorithms.

Our work raises
the question whether similar ``geometric'' versions are possible
for other, more advanced, network flow algorithms such as the one
by Goldberg and Rao~\cite{GoldbergRa99}. Similarly, it is an interesting
challenge to adapt algorithms for other combinatorial graph optimization
problems to the geometric setting. For a recent example that considers
the maximum matching problem, see~\cite{BonnetCaMu20}.

Finally, we cannot resist mentioning the tantalizing open problem
of settling the complexity status of the general barrier resilience
problem~\cite{kumar2007barrier}. Unlike for the strip version, 
we do not know any polynomial time algorithm for it. On the other hand,
up to now, all attempts at a proof of NP-hardness have failed.
An answer to this question would be most welcome.

\bibliographystyle{alpha}
\bibliography{bibliodisks}

\end{document}